\documentclass[showpacs,amsmath,amssymb,twocolumn,prl,superscriptaddress]{revtex4-1}

\usepackage[dvips]{graphicx} 
\usepackage{amsfonts}
\usepackage{amssymb}
\usepackage{amscd}
\usepackage{amsmath}    
\usepackage{enumerate}
\usepackage{epsfig}
\usepackage{subfigure}
\usepackage{xcolor}
\usepackage{amsthm}
\usepackage{physics}
\usepackage{tcolorbox}

\newtheorem{theorem}{Theorem}
\newtheorem{lemma}{Lemma}

\begin{document}
\preprint{APS/123-QED}
\title{Simultaneous Certification of Entangled States and Measurements in Bounded Dimensional Semi-Quantum Games}
\date{\today}
\author{Xingjian Zhang}
\affiliation{Center for Quantum Information, Institute for Interdisciplinary Information Sciences, Tsinghua University, Beijing, 100084 China}
\author{Qi Zhao}
\email{zhaoqithu10@gmail.com}
\affiliation{Center for Quantum Information, Institute for Interdisciplinary Information Sciences, Tsinghua University, Beijing, 100084 China}
\affiliation{Shanghai Branch, National Laboratory for Physical Sciences at Microscale and Department of Modern Physics, University of Science and Technology of China, Shanghai 201315, China}
\affiliation{CAS Center for Excellence and Synergetic Innovation Center in Quantum Information and Quantum Physics, University of Science and Technology of China, Shanghai 201315, China}

\begin{abstract}
  Certification of quantum systems and operations is a central task in quantum information processing. Most current schemes rely on a tomography with fully characterised devices, while this may not be met in real experiments. Device characterisations can be removed with device-independent tests, it is technically challenging at the moment, though. In this letter, we investigate the problem of certifying entangled states and measurements via semi-quantum games, a type of non-local quantum games with well characterised quantum inputs, balancing practicality and device-independence. We first design a specific bounded-dimensional measurement-device-independent game, with which we simultaneously certify any pure entangled state and Bell state measurement operators. Afterwards via a duality treatment of state and measurement, we interpret the dual form of this game as a source-independent bounded-dimensional entanglement swapping protocol and show the whole process, including any entangled projector and Bell states, can be certified with this protocol. In particular, our results do not require a complete Bell state measurement, which is beneficial for experiments and practical use.
\end{abstract}

\maketitle

\emph{Introduction.---}
With the rapid development of quantum technologies in state preparation and dynamical evolution, quantum devices are expected to outperform their classical counterparts in the future. However, unexpected and uncharacteristic noise hampers the functioning of quantum devices. Thus certification of states and processes which have been actually implemented in the quantum devices is a central task in this field.
Nonetheless, as a classical being, an observer can only take advantage of classical statistics together with a physical model he/she believes to characterise the system and its dynamical evolution. One way for certification is to employ a full tomography on the system (state tomography)~\cite{raymer1994complex,leonhardt1996discrete,leonhardt1997measuring} or on the process (process tomography)~\cite{poyatos1997complete,chuang1997prescription,d2001quantum}. With an information-complete set of measurements (input states), one can reconstruct all the elements in the quantum state operator (quantum process).
However, such approaches require a full knowledge on the devices. 
Unavoidable noise from the real environment can easily nullify the results. The states and measurements used in tomography need to be fully characterised and trusted, while sometimes they might be even more difficult to calibrate than the system/process investigated.

Luckily, quantum physics allows us to break away from the dilemma by Bell nonlocality~\cite{Bell,clauser1969proposed}. Violation of a Bell inequality indicates entanglement in a nonlocal system and incompatibility of measurements. In fact, as first explicitly pointed out by Mayers and Yao ~\cite{mayers2003self}, the state and measurements can be uniquely determined up
to local isometries from certain classical correlations, or, the system ``self-tests'' itself~\cite{mayers2003self} without any charactisation of the devices, but
only relying on the validity of quantum physics. In this sense we say the certification is device-independent.
So far, numerous remarkable results have been derived in the subject of state self-testing~\cite{yang2013robust,coladangelo2017all,mckague2011self,pal2014device,wu2014robust,vsupic2018self}, and physical processes or measurements certification~\cite{mckague2010generalized,bamps2015sum,vsupic2016self,kaniewski2017self,bancal2018noise,kaniewski2018maximal,sekatski2018certifying}. To make self-testing meet a real situation, robust self-testing, allowing the noises and imperfections to some extent,  has been investigated, aiming to give a lower bound on the fidelity of the physical system from a reference system (in the sense of local isometries) based on the observed statistics. Much progress has been made in this field, too~\cite{mckague2012robust,bancal2015physical,kaniewski2016analytic,mckague2012robust,yang2014robust,wu2014robust} .


Realising loophole-free device-independent tests are highly challenging in practice, mainly suffering from detection and locality loopholes.  Up till now, only a few device-independent experiments succeed in closing the detection loophole and the locality loophole simultaneously~\cite{giustina2015significant,hensen2015loophole,shalm2015strong,rosenfeld2017event,liu2018device}. However, in these demonstrations, either only a rather small violation of Bell inequality is obtained, which is far from certifying a Bell state according to the current robust self-testing results~\cite{kaniewski2016analytic}, or the repetition rate is too low for any practical applications.


A compromising approach to balance device-independence and practicality is to apply a semi-device-independent scenario. One of the most prominent scenarios is the semi-quantum game proposed by Buscemi~\cite{buscemi2012all}. Semi-quantum games are similar to Bell tests, except for that general local quantum inputs are allowed. It has been proved that all entangled states exhibit non-locality in these games, bridging the gap between the concept of entanglement and non-locality.
After Buscemi's seminal results, focusing on the prepared entangled states, there are following works presenting approaches to qualitatively witness entanglement~\cite{branciard2013measurement,xu2014implementation,verbanis2016resource,zhao2016efficient,shahandeh2017measurement} and quantitatively estimate the amount of entanglement in the system with semi-quantum games ~\cite{regev2015quantum,vsupic2017measurement}. However, little attention has been received for channels or measurements in these games, with one exceptional work focusing on the quantum channel with quantum memories ~\cite{rosset2018resource}.
Moreover, it is left open whether we can ``uniquely'' certify certain systems and operations simultaneously in a manner similar to self-testing in device-independent tests.


In this letter, we consider the certification of entangled states and measurements in semi-quantum games with a given dimensional Hilbert space. 
We design a new type of semi-quantum game, showing the plausibility of simultaneous certification of any pure bipartite entangled state in $\mathcal{H}^2\otimes\mathcal{H}^2$ and Bell state measurement operators. In particular, our design does not rely on a complete Bell state measurement but can be naturally generalised for such a case.
Moreover, we transform the setting into the manner of a source-independent entanglement swapping protocol via a duality treatment of state and measurement operator. In its ``dual'' form of semi-quantum game, any entangled projector acting on the bounded dimensional space and Bell states can be certify simultaneously.

\emph{Preliminary.---}We first briefly review the concept of semi-quantum games in~\cite{buscemi2012all} and results in~\cite{branciard2013measurement}. Consider a nonlocal game with two players, Alice and Bob, as shown in Fig.~\ref{Fig:SemiQ}. In each round of the game, a referee gives them quantum states $\psi_x^{A_0}\in\mathcal{D}(\mathcal{H}_{A_0}), \psi_y^{B_0}\in\mathcal{D}(\mathcal{H}_{B_0})$ individually. 
We denote the input systems as $\mathcal{H}_{A_0},\mathcal{H}_{B_0}$ for Alice and Bob, respectively. Then they are asked to give classical outputs $a,b$, and a score $\beta_{a,b}^{x,y}$  is obtained based on the combination of $\{\psi_x^{A_0},\psi_y^{B_0},a,b\}$.
In this game, generally, Alice and Bob can first share a bipartite state $\rho^{AB}\in\mathcal{D}(\mathcal{H}_A\otimes\mathcal{H}_B)$, and perform joint positive operator-valued measure measurements (POVMs) on the party they each hold and the received quantum input state, individually.
We denote the POVM elements as $M_a^{A_0A}$ acting on $\mathcal{H}_{A_0}\otimes\mathcal{H}_A$ yielding the result $a$ on Alice's side and $M_b^{BB_0}$ acting on $\mathcal{H}_B\otimes\mathcal{H}_{B_0}$ yielding the result $b$ on Bob's side. The average score is
\begin{equation}
\begin{aligned}
    S &= \sum_{x,y,a,b}\beta_{a,b}^{x,y}P(a,b|\psi_x^{A_0},\psi_y^{B_0}) \\
      &= \sum_{x,y,a,b}\beta_{a,b}^{x,y}\mathrm{Tr}[(\psi_x^{A_0}\otimes\rho^{AB}\otimes\psi_y^{B_0})(M_a^{A_0A} \otimes M_b^{BB_0})].
\end{aligned}
\end{equation}
\begin{figure}[!hbt]
\centering \resizebox{7cm}{!}{\includegraphics{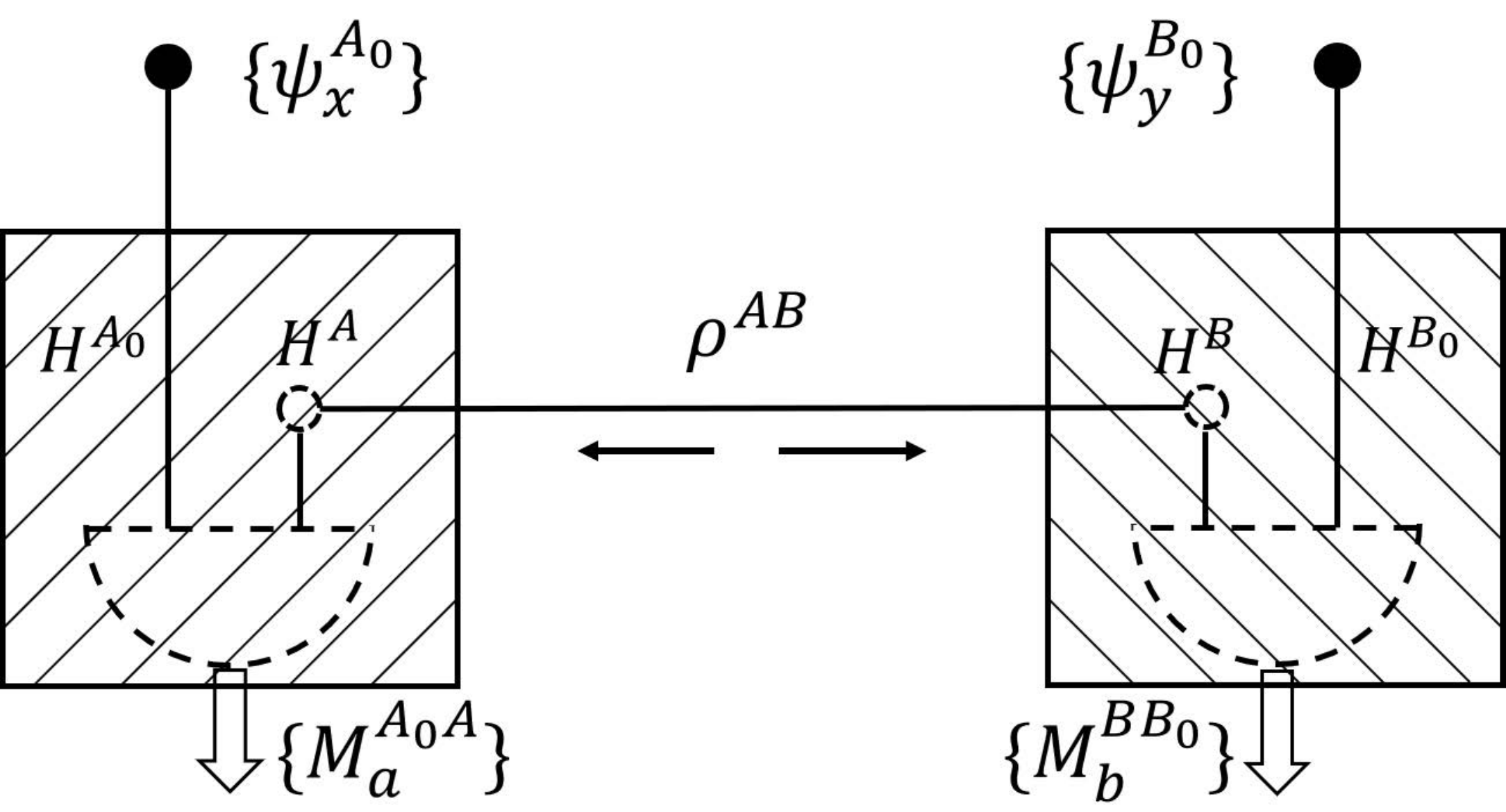}}
\caption{The semi-quantum game. The systems are denoted by their corresponding Hilbert spaces. In this game, from the referee's perspective, well-characteristic states $\{\psi_x^{A_0}\},\{\psi_y^{B_0}\}$ are sent to Alice and Bob, respectively, where each of the two players measures the input quantum state received and her/his own party of $\rho^{AB}$ jointly. The joint measurements are expressed as POVMs $\{M_{a}^{A_0A}\},\{M_{b}^{B_0A}\}$.} \label{Fig:SemiQ}
\end{figure}


Using the technique of partial POVM, we can treat inputting quantum states as a tomography process of the effective POVM elements
\begin{equation}\label{partial POVM}
    \tilde{M}_{ab}^{A_0B_0}=\mathrm{Tr}_{AB}\left[\left(I^{A_0}\otimes\rho^{AB}\otimes I^{B_0}\right)\left(M_a^{A_0A}\otimes M_b^{BB_0}\right)\right].
\end{equation}
It has been proved in~\cite{branciard2013measurement} that a separable $\rho^{AB}$ will result in separable effective POVM elements $\tilde{M}_{ab}^{A_0B_0}$ whatever the measurements. Ideally, an entangled $\rho^{AB}$ together with the projections on the maximally entangled state $M_a^{A_0A} = M_b^{BB_0} = \ketbra{\Phi^{+}}$ gives rise to the effective POVM element
$\tilde{M}_{ab}^{A_0B_0} = (\rho^{AB})^T/d_A d_B,$ where $d_A,d_B$ are dimensions of $\mathcal{H}_A,\mathcal{H}_B$ and $\ket{\Phi^{+}} = \dfrac{1}{\sqrt{d}}\sum_{i=1}^{d}\ket{ii}$. One can focus on one specific outcome, say $(a,b)=(0,0)$, by letting $\beta_{a,b}^{x,y}=\delta_{a,0}\cdot\delta_{b,0}\cdot\beta^{x,y}$ in a semi-quantum game.

Consider any conventional entanglement witness $W$, $\Tr[W\rho^{AB}]<0 $  for a specific entangled state $\rho^{AB}$ while $\Tr[W\sigma]\ge 0$ for any separable operator $\sigma$.
One can transform the witness $W$  into a semi-quantum game with the parameters $\beta^{x,y}$ and quantum inputs $\{\psi_x^{A_0},\psi_y^{B_0}\}$  in a decomposition of $W$,
\begin{equation}\label{MDI Bell operator}
    W = \sum_{x,y}\beta^{x,y}\psi_x^{A_0}\otimes\psi_y^{B_0},
\end{equation}
and a negative score $S<0$ implies the entanglement in systems $A,B$.
Furthermore, these semi-quantum game can also be applied to detect detailed entanglement structure \cite{zhao2016efficient}, estimate the robustness of entanglement, negativity and randomness \cite{vsupic2017measurement}.


\emph{Simultaneous certification of entangled state and measurements.---}
In previous works, only prepared states or channels are characterized via a semi-quantum game.
To further elaborate the application of semi-quantum games, we now investigate the problem of certifying entangled states and joint measurements simultaneously. We first define the concept of entangled measurement operators and entangled measurements:


\emph{Definition 1 (Entangled Measurements).---}A measurement operator $M_{ab}^{AB}$ acting on the system $\mathcal{H}_A\otimes\mathcal{H}_B$ is called entangled if it cannot be expressed as $\sum_{i,j}N_{i}^A\otimes N_{j}^B$ with a set of semi-definite positive operators $\{N_{i}^A,N_{j}^B|N_{i}^A\in\mathcal{L}(\mathcal{H}_A), N_{j}^B\in\mathcal{L}(\mathcal{H}_B),N_{i}^A,N_{j}^B\succeq 0\}$, otherwise it is called a separable measurement operator. A POVM $\{M_{ab}^{AB}\}$ is called entangled if at least one of its measurement operators is entangled.

Entangled measurements, especially  Bell state measurement (BSM), play an important role in many quantum information processing tasks like entanglement swapping and teleportation.
In the measurement-device-independent protocols, entangled measurement operators
are essential for a nonlocal behaviour, too.
we have the following result for semi-quantum games:

\emph{Theorem 1.---}In a semi-quantum game where the input states and appointed scores form an entanglement witness given by Eq.~\eqref{MDI Bell operator}, an observation of an average score $S=\Tr[W\tilde{M}_{0,0}^{A_0B_0}]<0$ necessarily indicates that $M_a^{A_0A}$ is an entangled POVM operator on systems $A_0$ and $A$, and the same property holds for $M_b^{BB_0}$.


The proof is given in Appendix.
This theorem together with previous results show that entangled measurements and state can be witnessed by investigating the effective POVM elements. Now we take one step further and try to certify them via a quantitative investigation on the effective POVM measurement operators $\tilde{M}_{ab}^{A_0B_0}$. We consider the case with a bounded dimensional Hilbert space, specifically, the case where the unknown state is a pair of qubits, and we design the semi-quantum games such that all quantum inputs are qubits as well, i.e. $\mathcal{H}_A \cong\mathcal{H}_B \cong\mathcal{H}_{A_0} \cong\mathcal{H}_{B_0} \cong\mathcal{H}^2$.

We follow the design of measurement-device-independent entanglement witness and also focus on one pair of measurement operators corresponding to a single outcome, say $(a,b)=(0,0)$. For brevity, we omit the subscripts and denote the measurement operators as $M^{A_0A}, M^{BB_0}$ and the effective POVM element as $\tilde{M}^{A_0B_0}$.
In Eq.~\eqref{MDI Bell operator}, we choose an operator $W$ with the spectral decomposition $W = \sum_{i=1}^4 l_i \ketbra{\psi_i}$
and design a following measurement-device-independent game:

\begin{tcolorbox}[title = {Design of the measurement-device-independent game:}]
Game $\{\psi_x^{A_0},\psi_y^{B_0},\beta_{0,0}^{x,y}\}$ is designed such that
\begin{equation}\nonumber
\begin{aligned}
&W = \sum_{i=1}^4 l_i \ketbra{\psi_i}= \sum_{x,y}\beta_{0,0}^{x,y}\psi_x^{A_0}\otimes\psi_y^{B_0},\\
&l_1 > 0 > l_2 \geq l_3 \geq l_4,\,|l_2|,|l_3|,|l_4| \gg l_1, \\
& \forall i, j, \braket{\psi_i}{\psi_j}=\delta_{i,j},\\
&\ket{\psi_1} \text{ is an entangled pure state}.\\
\end{aligned}
\end{equation}
\end{tcolorbox}
Now we give the main theorem in this letter:

\emph{Theorem 2.---}Under the condition that $\mathcal{H}_A \cong\mathcal{H}_B \cong\mathcal{H}_{A_0} \cong\mathcal{H}_{B_0} \cong\mathcal{H}^2$,
any two-qubit entangled pure state $\ket{\psi_1}$ and corresponding joint measurement operators $M^{A_0A} $, $M^{BB_0}$ can be certified, with the largest score  $l_1/4$ in the above game obtained. To be more specific, the observation of this score certifies that $(\rho^{AB})^T\sim \ketbra{\psi_1},\, M^{A_0A}, M^{BB_0} \sim \ketbra{\Phi^+}$, where the equivalence refers to a freedom of local unitary operations.

\emph{A sketch of proof.---}Our design of the semi-quantum game restricts that $\tilde{M}^{A_0B_0}$ must lie in the support of $\ketbra{\psi_1},\,\text{supp}\{\ketbra{\psi_1}\}$, in order to maximize the average score $S$. This allows us to represent the problem in an optimization form:
\begin{equation}\label{optimization}\nonumber
\begin{aligned}
    &\arg\max_{\rho^{AB},M_0^{A_0A},M_0^{BB_0}}\Tr[\tilde{M}^{A_0B_0}\ketbra{\psi_1}], \\
    &\textbf{Subject to:} \\
    &\tilde{M}^{A_0B_0}=\mathrm{Tr}_{AB}\left[\left(I^{A_0}\otimes\rho^{AB}\otimes I^{B_0}\right)\left(M_0^{A_0A}\otimes M_0^{BB_0}\right)\right], \\
    &\Tr[\tilde{M}^{A_0B_0}(I^{A_0B_0}-\ketbra{\psi_1})] = 0.
\end{aligned}
\end{equation}
Here $\arg$ refers to the arguments when the maximum of the objective function is taken. In order to maximize the objective function, $M^{A_0A},M^{BB_0}$ need to be projectors. Yet we do not make any assumptions on their ranks. Furthermore, the state $\rho^{AB}$ is in general a mixed state. However, we have the following lemmas showing that only rank-one measurement operators and a pure state can guarantee that $\tilde{M}^{A_0B_0}$ is projected on the span of a pure entangled state:

\emph{Lemma 1.---}For $\mathcal{H}_A \cong\mathcal{H}_B \cong\mathcal{H}_{A_0} \cong\mathcal{H}_{B_0} \cong\mathcal{H}^2$ and $M^{A_0A},M^{BB_0}$ rank-one projectors, a necessary condition for $\tilde{M}^{A_0B_0}$ to be a rank-one operator is that $\rho^{AB}$ is a pure state.

\emph{Lemma 2.---}For a pure quantum state $\rho^{AB}=\ketbra{\Psi}\in\mathcal{D}(\mathcal{H}^2\otimes\mathcal{H}^2)$, a necessary condition for $\tilde{M}^{A_0B_0}$ to be a rank-one operator is that both $M^{A_0A}$ and $M^{BB_0}$ are rank-one operators.

We leave the proofs in Appendix. For a general state $\rho^{AB}$ and POVM elements $M^{A_0A},\,M^{BB_0}$, we can apply spectral decomposition to them,
\begin{equation}\nonumber
\begin{aligned}
    \rho^{AB}&=\sum_i\lambda_i\ketbra{\Psi_i}, \\
    M^{A_0A} &= \sum_a\ketbra{\phi_a},\, M^{BB_0}=\sum_b\ketbra{\phi_b}.
\end{aligned}
\end{equation}
For fixed $a,b$, the operator
\begin{equation}\nonumber
\Tr_{AB}\left[\left(I^{A_0}\otimes\rho^{AB}\otimes I^{B_0}\right)\\\left(\ketbra{\phi_a}\otimes \ketbra{\phi_b}\right)\right]
\end{equation}
is positive, thus $\tilde{M}^{A_0B_0}$, the sum of above operators, is rank-1 only if these operators all project on $\text{supp}\{\ketbra{\psi_1}\}$. Via Lemma 1, we relax $\rho^{AB}$ to be pure. Then with a pure $\rho^{AB}$, we require $M^{A_0A},\,M^{BB_0}$ both to be rank-1 operators so that $\tilde{M}^{A_0B_0}\in\text{supp}\{\ketbra{\psi_1}\}$.

We now can focus on the case with rank-one projectors and a pure state.
First we apply the Schmidt decomposition to $\ket{\psi_1}$, and without loss of generality it can be written as $\ket{\psi_1} = \cos{\chi}\ket{00} + \sin{\chi}\ket{11}$,
where $\ket{0},\ket{1}$ represent orthogonal bases in systems $A_0,B_0$. Since the measurement operators can be treated as pure bipartite states now, we now denote them as $M^{A_0A}=\ketbra{\phi^{A_0A}}, M^{BB_0}=\ketbra{\phi^{BB_0}}$. Furthermore, we express them in the form of
\begin{equation}\label{rank-1 M}
\begin{aligned}
    \ket{\phi^{A_0A}} = \cos{\alpha}\ket{0} \ket{\varphi^A} + \sin{\alpha}\ket{1} \ket{\bar{\varphi}^A}, \\
    \ket{\phi^{BB_0}} = \cos{\beta}\ket{0} \ket{\varphi^B} + \sin{\beta}\ket{1} \ket{\bar{\varphi}^B}.
\end{aligned}
\end{equation}
In Eq.~\eqref{rank-1 M}, the state vectors $\ket{0},\ket{1}$ acting on systems $A_0,B_0$ are the same as the ones in representing $\ket{\psi_1}$. The vectors $\ket{\varphi^A},\ket{\bar{\varphi}^A}\in\mathcal{H}_A, \, \ket{\varphi^B},\ket{\bar{\varphi}^B}\in\mathcal{H}_B$ are normalized, yet in general $\bra{\varphi^A}\ket{\bar{\varphi}^A},\bra{\varphi^B}\ket{\bar{\varphi}^B} \neq 0$. Therefore, we can express $\tilde{M}^{A_0B_0} = \ketbra{\tilde{\Psi}}$ where $\ket{\tilde{\Psi}}$ is a sub-normalized vector
\begin{equation}\nonumber
\begin{aligned}
&\cos{\alpha}\cos{\beta}\bra{\Psi}\ket{\varphi^A\varphi^B}\ket{00} +\sin{\alpha}\sin{\beta}\bra{\Psi}\ket{\bar{\varphi}^A\bar{\varphi}^B}\ket{11} \\ +&\cos{\alpha}\sin{\beta}\bra{\Psi}\ket{\varphi^A\bar{\varphi}^B}\ket{01} +\sin{\alpha}\cos{\beta}\bra{\Psi}\ket{\bar{\varphi}^A\varphi^B}\ket{10}.
\end{aligned}
\end{equation}
Requiring $\ket{\tilde{\Psi}}$ to be parallel to $\ket{\psi_1}$ indicates that $\bra{\Psi}\ket{\varphi^A\bar{\varphi}^B} = \bra{\Psi}\ket{\bar{\varphi}^A\varphi^B}=0$. Under this condition we have the following lemma:

\emph{Lemma 3.---}In the condition that $\bra{\Psi}\ket{\varphi^A\bar{\varphi}^B} = \bra{\Psi}\ket{\bar{\varphi}^A\varphi^B}=0$, where $\ket{\varphi^A},\ket{\bar{\varphi}^A},\ket{\varphi^B},\ket{\bar{\varphi}^B}$ are introduced by Eq.~\eqref{rank-1 M}, we have $\left|\bra{\Psi}\ket{\varphi^A\varphi^B}\right|^2+\left|\bra{\Psi}\ket{\bar{\varphi}^A\bar{\varphi}^B}\right|^2 \leq 1$. The equality is taken iff $\bra{\varphi^A}\ket{\bar{\varphi}^A}=\bra{\varphi^B}\ket{\bar{\varphi}^B}=0$, and the subspaces of the two parties are spanned by $\left\{\ket{\varphi^A},\,\ket{\bar{\varphi}^A}\right\},\, \left\{\ket{\varphi^B},\,\ket{\bar{\varphi}^B}\right\}$, respectively.

While the lemma is obvious in the bounded dimension Hilbert space we now consider, it is worth noting that it does not rely on a restriction of dimensions of $\mathcal{H}_A,\mathcal{H}_B$. We leave the proof for this lemma in the general case in Appendix.

With this lemma and the requirement that $\ket{\tilde{\Psi}}$ lying parallel to $\ket{\psi_1}$, we can derive that
\begin{equation}
\begin{aligned}
    &\Tr[\tilde{M}^{A_0B_0}\ket{\psi_1}\bra{\psi_1}] \\ \leq &\dfrac{\sin^2{\alpha}\sin^2{\beta}\cos^2{\alpha}\cos^2{\beta}}{\sin^2{\alpha}\sin^2{\beta}\cos^2{\chi} + \cos^2{\alpha}\cos^2{\beta}\sin^2{\chi}}.
\end{aligned}
\end{equation}
The equality is reached when $|\bra{\Psi}\ket{\varphi^A\varphi^B}|^2 + |\bra{\Psi}\ket{\bar{\varphi}^A\bar{\varphi}^B}|^2 = 1$. The maximum of this expression is $1/4$, taken when $\sin^2{\alpha} = \sin^2{\beta} = 1/2$, which corresponds to Bell state measurement operators, and $\bra{\Psi}\ket{\varphi^A\varphi^B}/\bra{\Psi}\ket{\bar{\varphi}^A\bar{\varphi}^B} = \cos{\chi}/\sin{\chi}$, yielding the conclusion that $\ket{\Psi}$ is equivalent to $\ket{\psi_1}$. This finishes our proof.

\emph{Certification in a source-independent entanglement swapping game.---}
It is interesting to notice that the entanglement swapping protocol shares a dual form of the semi-quantum game if we exchange the roles of state and measurements, shown in
Fig.~\ref{Fig:EntSwap}. We have two independent bipartite states, $\rho^{A_0A}\in\mathcal{D}(\mathcal{H}_{A_0}\otimes\mathcal{H}_{A}), \rho^{BB_0}\in\mathcal{D}(\mathcal{H}_{B}\otimes\mathcal{H}_{B_0})$.
Different from the standard entanglement swapping protocol, here we consider a scenario with uncharacteristic states $\rho^{A_0A},\rho^{BB_0}$ and an uncharacteristic joint measurement $\{M_i^{AB}\}$ acting on systems $A,B$. Compared to the semi-quantum game in Fig.~1 where trusted quantum inputs are utilised, suppose now we can carry out well-characterised local measurements $\{N_a^{A_0}\},\{N_b^{B_0}\}$ on systems $A_0,B_0$. Thus our entanglement swapping game is similar to a (non)bilocal scenario \cite{PhysRevLett.104.170401,PhysRevA.85.032119}, except for that we can now apply trusted measurements on the systems $A_0,B_0$. Because the uncharacteristic underlying state $\rho^{A_0A}\otimes\rho^{BB_0}$, we say the protocol shares a source-independent nature \cite{cao2016source}.
\begin{figure}[!hbt]
\centering \resizebox{7cm}{!}{\includegraphics{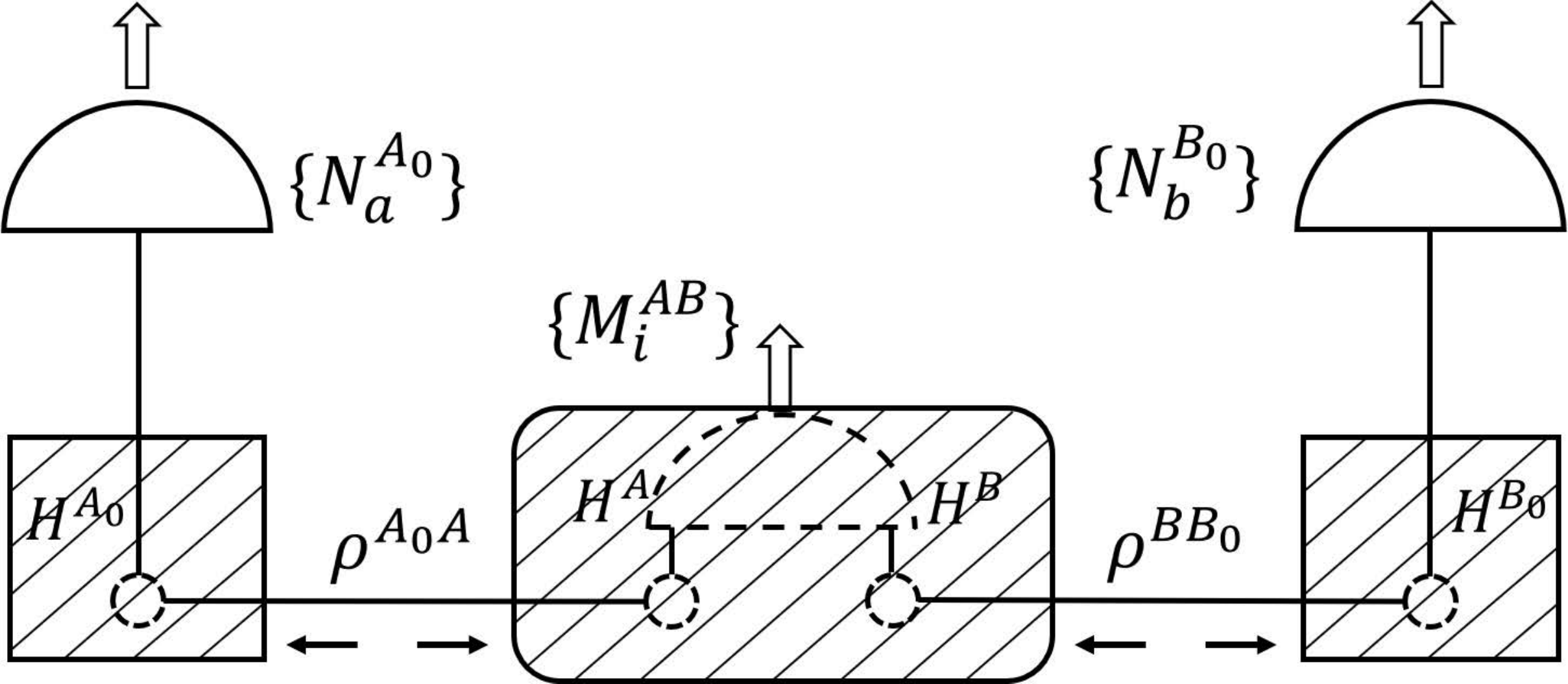}}
\caption{The uncharacteristic entanglement swapping protocol. In this protocol, two uncharacteristic states $\rho^{A_0A}\in\mathcal{D}(\mathcal{H}_{A_0}\otimes\mathcal{H}_{A}),\rho^{BB_0}\in\mathcal{D}(\mathcal{H}_{B}\otimes\mathcal{H}_{B_0})$ are prepared. No information is known about the two states except for that they are both a pair of qubits. A so-called Bell state measurement is performed on the systems $\mathcal{H}_{A},\mathcal{H}_{B}$. One can use well-characteristic local measurements $\{N_a^{A_0}\},\{N_b^{B_0}\}$ on systems $A_0,B_0$ to certify the results.} \label{Fig:EntSwap}
\end{figure}

Similar to the treatment of effective POVM operators, we can view the protocol as first preparing a sub-normalised positive semi-definite operator on $A_0,B_0$
\begin{equation}
    \tilde{\rho}_{i}^{A_0B_0} = \Tr_{AB}\left[\left(I^{A_0}\otimes M_i^{AB}\otimes I^{B_0}\right)\left(\rho^{A_0A}\otimes \rho^{BB_0}\right)\right]
\end{equation}
and measuring the systems $A_0,B_0$ subsequently. The probability of measurement outcome $i$ on systems $A,B$ is $\Tr[\tilde{\rho}_{i}^{A_0B_0}]$ and the corresponding state of systems $A_0,B_0$ becomes $\tilde{\rho}_{i}^{A_0B_0}/\Tr[\tilde{\rho}_{i}^{A_0B_0}]$. If we focus on one specific measurement result on systems $A,B$, with a duality treatment on states and measurements, the source-independent entanglement swapping protocol can be transformed into the measurement-device-independent-type semi-quantum game, except for a difference in the normalisation factors of states and measurement operators. Therefore, we may construct a dual semi-quantum game to certify the initial systems $A_0,A$ and $B,B_0$ and joint measurement operator simultaneously, with measurements on the final systems $A_0,B_0$ only.




With a set of tomographically-complete measurements $\{N_a^{A_0}\},\{N_b^{B_0}\}$ we can certify the final state of systems $A_0,B_0$.
In a dual semi-quantum game, a score $\beta_i^{a,b}$ is appointed to Alice and Bob based on the measurement results of $\{N_a^{A_0}\},\{N_b^{B_0}\},\{M_i^{AB}\}$, and the average score Alice and Bob can get is
\begin{equation}
\begin{aligned}
    S 
      = \sum_{i,a,b}\beta_{i}^{a,b}\Tr[(N_a^{A_0}\otimes M_i^{AB}\otimes N_b^{B_0})(\rho^{A_0A} \otimes \rho^{BB_0})].
\end{aligned}
\end{equation}
We now design a specific game as follows:

\begin{tcolorbox}[title = {Design of the source-independent entanglement swapping game:}]
Game $\{N_a^{A_0},N_b^{B_0},\beta_{i}^{a,b}\}$ is designed such that
\begin{equation}\nonumber
\begin{aligned}
V &= \sum_{a,b,i}\beta_{i}^{a,b}N_a^{A_0}\otimes N_b^{B_0} = \delta_{i,1}\sum_{a,b}\beta^{a,b}N_a^{A_0}\otimes N_b^{B_0}\\ &= \delta_{i,1}\left[\ketbra{\psi}-l\left(I-\ketbra{\psi}\right)\right], \\&\ketbra{\psi}\text{ is an entangled projector, } l\gg0.
\end{aligned}
\end{equation}
\end{tcolorbox}


In this game, the average score can be expressed as $S = \Tr[V\tilde{\rho}_1^{A_0B_0}]$. In order to maximize the score, the operator $\rho_1^{A_0B_0}$ needs to be embedded in the support of $\ketbra{\psi}$, which is much the same as the idea in the measurement-device-independent game. In our proofs for Lemma 1 and Lemma 2, we do not rely on the normalisation factors, but mainly the orthogonality between different eigenvectors in the spectral decomposition. Hence with a similar route, one can come to the following theorem:

\emph{Theorem 3.---}In the above source-independent entanglement swapping game with the assumption that all systems are two-dimensional, the largest score that can be obtained is $1/4$, and it is obtained only if $(M_1^{AB})^T = \ketbra{\psi}, \rho^{A_0A}, \rho^{BB_0} = \ket{\Phi^+}\bra{\Phi^+}$, up to local unitary operations.

This theorem actually states that an entangled projector and Bell states can be certified in this game, if the state $\tilde{\rho}_i^{A_0B_0}/\Tr[\tilde{\rho}_i^{A_0B_0}]$ prepared on the systems $A_0,B_0$ is a pure entangled state with a probability of $1/4$. We can use the result to certify the case with multiple outputs of the joint measurement, by observing all sub-normalised operators $\{\tilde{\rho}_i^{A_0B_0}\}$. One special case is that we can certify an ideal entanglement swapping process corresponding to a complete Bell state measurement:

\emph{Corollary.---}Under the qubit system assumption, if all measurement outcomes occur with a probability of $1/4$ and the corresponding final states on the systems $A_0,B_0$ form a complete set of Bell states under a certain choice of local bases, one can certify that the joint measurement $\{M_i^{AB}\}$ is a set of complete Bell state measurements, and $\rho^{A_0A},\,\rho^{BB_0}$ are Bell states.

%

\emph{Conclusions.---} Via focusing on the effective POVM operator seen from the input ports in a semi-quantum game, we show that any pure entangled state and Bell state measurement operators can be certified up to local unitary operations. Moreover,  we present the certification of Bell states and an entangled joint measurement operator in a source-independent entanglement swapping protocol, following a similar technique. This technique can also be expected to to certify other type of semi-quantum game, for instance, certification of  a quantum memory channel \cite{rosset2018resource}.

Due to a dimension restriction, we cannot treat our certification as a ``self-testing'' result. However, we conjecture that the same certification result can hold even if we do not restrict on the unknown system's dimension in the manner of self-testing. Lemma 3 in our approach relies only on the ranks of the measurement operators and quantum system rather than the system's dimension, and it may be a starting point for self-testing. Furthermore, robust semi-device-independent self-testing results can also be expected. It is also worth noting that our methods do not require a complete Bell state measurement result, which is often not easy to achieve in practice (in particular, it has been proved impossible to carry out a complete Bell state measurement in linear optics).
A robust semi-device-independent self-testing may be beneficial for practical blind quantum computing tasks \cite{reichardt2013classical,PhysRevLett.119.050503}. For future directions, it is also interesting to extend semi-quantum games into the non-i.i.d. region and parallel setting. We hope our work can shed light on the further exploration of applications of semi-quantum non-local games. 


\emph{Acknowledgement.---}
The authors would like to thank D. Cavalcanti, X. Yuan, X. Ma, Y. Liang, and Z. Wei for enlightening discussions. This work was supported by the National Natural Science Foundation of China Grants No.~11875173 and No.~11674193, the National Key R\&D Program of China Grants No.~2017YFA0303900 and No.~2017YFA0304004, and the Zhongguancun Haihua Institute for Frontier Information Technology.


\bibliographystyle{apsrev4-1}

\bibliography{bibMDI}

\onecolumngrid
\appendix

\section{Proof of Theorem 1}
For convenience, we first restate the theorem in Main text:
\begin{theorem}
In a semi-quantum game where the input states and appointed scores form an entanglement witness, an observation of an average score $S=\Tr[W\tilde{M}_{0,0}^{A_0B_0}]<0$ necessarily indicates that $M_a^{A_0A}$ is an entangled POVM operator on systems $A_0$ and $A$, and the same property holds for $M_b^{BB_0}$.
\end{theorem}

\begin{proof}
Suppose the POVM operator $M_a^{A_0A}=\sum_{\lambda_1,\lambda_2}M_{\lambda_1}^{A_0}\otimes M_{\lambda_2}^{A},M_{\lambda_1}^{A_0}, M_{\lambda_2}^{A}\geq0,\,\forall\lambda_1,\lambda_2$, is a separable operator between $A_0,A$. We do not make constraints on the POVM operator $M_b^{BB_0}$ and the state $\rho^{AB}$. Then the sub-normalized operator $\tilde{M}_{ab}^{A_0B_0}$ is
\begin{equation}
\begin{aligned}
    \tilde{M}_{ab}^{A_0B_0} &= \Tr_{AB}\left[\left(I^{A_0}\otimes\rho^{AB}\otimes I^{B_0}\right)\left(M_a^{A_0A}\otimes M_b^{BB_0}\right)\right] \\
    &= \mathrm{Tr}_{AB}\left[\left(I^{A_0}\otimes\rho^{AB}\otimes I^{B_0}\right)\left[\left(\sum_{\lambda_1,\lambda_2}M_{\lambda_1}^{A_0}\otimes M_{\lambda_2}^{A}\right)\otimes M_b^{BB_0}\right)\right] \\
    &=\sum_{\lambda_1,\lambda_2}M_{\lambda_1}^{A_0}\otimes\Tr_{A}\left(M_{\lambda_2}^{A}\cdot\Tr_{B}\left[\left(I^{A}\otimes M_b^{BB_0})(\rho^{AB}\otimes I^{B_0}\right)\right]\right),
\end{aligned}
\end{equation}
which is a separable operator acting on $\mathcal{H}_{A_0},\mathcal{H}_{B_0}$. Similarly, $\tilde{M}_{ab}^{A_0B_0}$ is separable if $M_b^{BB_0}$ is separable on $B,B_0$.

From a qualitative perspective, we can conclude that an entangled $\tilde{M}_{ab}^{A_0B_0}$ not only indicates the entanglement of $\rho^{AB}$, but also an entangled structure of the measurement operators $M_a^{A_0A},M_b^{BB_0}$.
\end{proof}

\section{Proof of Lemma 1}
\begin{lemma}
For $\mathcal{H}_A \cong\mathcal{H}_B \cong\mathcal{H}_{A_0} \cong\mathcal{H}_{B_0} \cong\mathcal{H}^2$ and $M^{A_0A},M^{BB_0}$ rank-one projectors, in order that $\tilde{M}_{ab}^{A_0B_0}=\Tr_{AB}\left[\left(I^{A_0}\otimes\rho^{AB}\otimes I^{B_0}\right)\left(M_a^{A_0A}\otimes M_b^{BB_0}\right)\right]$ is a rank-1 operator, $\rho^{AB}$ needs to be a pure state.
\end{lemma}

\begin{proof}
We apply the spectra decomposition on $\rho^{AB}$
\begin{equation}
    \rho^{AB} = \sum_{i=1}^4\mu_i\ket{\Psi_i}\bra{\Psi_i},
\end{equation}
where $\sum_{i=1}^4\mu_i=1,\mu_i\geq0,\forall i$.
We have
\begin{equation}
\begin{aligned}
    \tilde{M}^{A_0B_0} &= \mathrm{Tr}_{AB}\left[\left(I^{A_0}\otimes\rho^{AB}\otimes I^{B_0}\right)\left(M^{A_0A}\otimes M^{BB_0}\right)\right] \\
    &= \sum_i \mu_i \mathrm{Tr}_{AB}\left[\left(I^{A_0}\otimes\ket{\Psi_i}\bra{\Psi_i} \otimes I^{B_0}\right)\left(M^{A_0A}\otimes M^{BB_0}\right)\right] \\
    &\equiv\sum_i \mu_i \ket{\tilde{\Psi}_i}\bra{\tilde{\Psi}_i}.
\end{aligned}
\end{equation}
As pointed out in the main text, $M^{A_0A},M^{BB_0}$ can be taken as projectors $M^{A_0A}=\ket{\phi^{A_0A}}\bra{\phi^{A_0A}}, M^{BB_0}=\ket{\phi^{BB_0}}\bra{\phi^{BB_0}}$. With our restriction on the space dimension, without loss of generality we can express the vectors $\ket{\phi^{A_0A}},\ket{\phi^{BB_0}}$ using Schmidt decomposition
\begin{equation}\label{SM:POVM Schmidt}
\begin{aligned}
    \ket{\phi^{A_0A}} = \cos{\alpha}\ket{\varphi^{A_0}} \ket{\varphi^A} + \sin{\alpha}\ket{{\varphi}^{A_0\bot}} \ket{{\varphi}^{A\bot}}, \\
    \ket{\phi^{BB_0}} = \cos{\beta}\ket{\varphi^{B_0}} \ket{\varphi^B} + \sin{\beta}\ket{{\varphi}^{B_0\bot}} \ket{{\varphi}^{B\bot}},
\end{aligned}
\end{equation}
where $\bra{\varphi^A}\ket{{\varphi}^{A\bot}} = \bra{\varphi^B}\ket{{\varphi}^{B\bot}} = \bra{\varphi^{A_0}}\ket{{\varphi}^{A_0\bot}} = \bra{\varphi^{B_0}}\ket{{\varphi}^{B_0\bot}} = 0$. Then we have
\begin{equation}
\begin{aligned}
\ket{\tilde{\Psi_i}}=&\cos{\alpha}\cos{\beta}\bra{\Psi_i}\ket{\varphi^A\varphi^B}\ket{\varphi^{A_0}\varphi^{B_0}} + \cos{\alpha}\sin{\beta}\bra{\Psi_i}\ket{\varphi^A{\varphi}^{B\bot}}\ket{\varphi^{A_0}{\varphi}^{B_0\bot}} \\
        &+\sin{\alpha}\cos{\beta}\bra{\Psi_i}\ket{{\varphi}^{A\bot}\varphi^B}\ket{{\varphi}^{A_0\bot}\varphi^{B_0}}
        +\sin{\alpha}\sin{\beta}\bra{\Psi_i}\ket{{\varphi}^{A\bot}{\varphi}^{B\bot}}\ket{{\varphi}^{A_0\bot}{\varphi}^{B_0\bot}}.
\end{aligned}
\end{equation}
Now $\{\ket{\varphi^A},\ket{{\varphi}^{A\bot}}\}, \{\ket{\varphi^B},\ket{{\varphi}^{B\bot}}\}$ form complete bases for the corresponding spaces, therefore we have
\begin{equation}
    \left|\bra{\Psi_i}\ket{\varphi^A\varphi^B}\right|^2 + \left|\bra{\Psi_i}\ket{{\varphi}^{A\bot}{\varphi}^{B\bot}}\right|^2 + \left|\bra{\Psi_i}\ket{{\varphi}^{A\bot}\varphi^B}\right|^2 + \left|\bra{\Psi_i}\ket{\varphi^A{\varphi}^{B\bot}}\right|^2 = 1,\,\forall i.
\end{equation}
With this complete basis we can represent the eigenvectors of $\rho^{AB}$ as
\begin{equation}
    \ket{\Psi_i} = e^{i\theta_{i_1}}\cos{\mu_i}\cos{\gamma_i}\ket{\varphi^A\varphi^B} + e^{i\theta_{i_2}}\cos{\mu_i}\sin{\gamma_i}\ket{\varphi^A{\varphi}^{B\bot}} + e^{i\theta_{i_3}}\sin{\mu_i}\cos{\gamma_i}\ket{{\varphi}^{A\bot}\varphi^B} + e^{i\theta_{i_4}}\sin{\mu_i}\sin{\gamma_i}\ket{{\varphi}^{A\bot}{\varphi}^{B\bot}}.
\end{equation}
Hence
\begin{equation}
\begin{aligned}
\ket{\tilde{\Psi_i}}=&e^{i\theta_{i_1}}\cos{\mu_i}\cos{\gamma_i}\cos{\alpha}\cos{\beta}\ket{\varphi^{A_0}\varphi^{B_0}} +e^{i\theta_{i_2}}\cos{\mu_i}\sin{\gamma_i}\cos{\alpha}\sin{\beta}\ket{\varphi^{A_0}{\varphi}^{B_0\bot}} \\ &+e^{i\theta_{i_3}}\sin{\mu_i}\cos{\gamma_i}\sin{\alpha}\cos{\beta}\ket{{\varphi}^{A_0\bot}\varphi^{B_0}} + e^{i\theta_{i_4}}\sin{\mu_i}\sin{\gamma_i}\sin{\alpha}\sin{\beta}\ket{{\varphi}^{A_0\bot}{\varphi}^{B_0\bot}}. \\
\end{aligned}
\end{equation}
In order that the final operator $\tilde{\rho}^{A_0B_0}$ is rank-1, the only way is to guarantee that $\ket{\Psi_i}\sim\ket{\Psi_i},\,\forall i,j$. We can assume that the angle in the phase parameter $\theta_{i_k}=0$. Now we focus on two vectors, $\ket{\tilde{\Psi_1}},\ket{\tilde{\Psi_2}}$. By $\ket{\tilde{\Psi_1}}\sim\ket{\tilde{\Psi_2}}$ we require
\begin{equation}
    \cos{\mu_1}\cos{\gamma_1}:\cos{\mu_1}\sin{\gamma_1}:\sin{\mu_1}\cos{\gamma_1}:\sin{\mu_1}\sin{\gamma_1} = \cos{\mu_2}\cos{\gamma_2}:\cos{\mu_2}\sin{\gamma_2}:\sin{\mu_2}\cos{\gamma_2}:\sin{\mu_2}\sin{\gamma_2}.
\end{equation}
However, by spectral decomposition we have $\ket{\Psi_i}$ to be orthogonal to each other, which requires that
\begin{equation}
    \cos{\mu_1}\cos{\gamma_1}\cos{\mu_2}\cos{\gamma_2} + \cos{\mu_1}\sin{\gamma_1}\cos{\mu_2}\sin{\gamma_2} + \sin{\mu_1}\cos{\gamma_1}\sin{\mu_2}\cos{\gamma_2} + \sin{\mu_1}\sin{\gamma_1}\sin{\mu_2}\sin{\gamma_2} = 0.
\end{equation}
Obviously we cannot have a nontrivial result satisfying both requirements. Therefore, one cannot have a rank-1 $\tilde{\rho}^{A_0B_0}$ by measuring a mixed $\rho^{AB}$.
\end{proof}

\section{Proof of Lemma 2}
\begin{lemma}
For $\mathcal{H}_A \cong\mathcal{H}_B \cong\mathcal{H}_{A_0} \cong\mathcal{H}_{B_0} \cong\mathcal{H}^2$ and a pure quantum state $\rho^{AB}=\ket{\Psi}\bra{\Psi}\in\mathcal{D}(\mathcal{H}_A\otimes\mathcal{H}_B)$, in order that the effective POVM element $\tilde{M}^{A_0B_0}=\Tr_{AB}\left[\left(I^{A_0}\otimes\rho^{AB}\otimes I^{B_0}\right)\left(M_a^{A_0A}\otimes M_b^{BB_0}\right)\right]$ is a rank-one operator, $M^{A_0A}$ needs to be a rank-one operator. The same holds for $M^{BB_0}$.
\end{lemma}

\begin{proof}
We assume that the POVM element on Bob's side is a rank-1 projector, $M^{BB_0}=\ket{\phi^{BB_0}}\bra{\phi^{BB_0}}$, while the element on Alice's side is not extreme. We now focus on two eigenvectors of $M^{A_0A}, \ket{\phi_1^{A_0A}},\ket{\phi_2^{A_0A}}$. We apply the Schmidt decomposition on these vectors and without loss of generality we write
\begin{equation}
\begin{aligned}
    &\ket{\phi_1^{A_0A}} = \cos{\alpha_1}\ket{\varphi_1^{A_0}} \ket{\varphi_1^A} + \sin{\alpha_1}\ket{{\varphi}_1^{A_0\bot}} \ket{{\varphi}_1^{A\bot}}, \\
    &\ket{\phi_2^{A_0A}} = \cos{\alpha_2}\ket{\varphi_2^{A_0}} \ket{\varphi_2^A} + \sin{\alpha_2}\ket{{\varphi}_2^{A_0\bot}} \ket{{\varphi}_2^{A\bot}}, \\
    &\ket{\phi^{BB_0}} = \cos{\beta}\ket{\varphi^{B_0}} \ket{\varphi^B} + \sin{\beta}\ket{{\varphi}^{B_0\bot}} \ket{{\varphi}^{B\bot}}.
\end{aligned}
\end{equation}
Using extra degrees of freedom, let
\begin{equation}
\begin{aligned}
    &\left\{
        \begin{array}{lr}
          \ket{\varphi_2^{A_0}} = \cos{\gamma}\ket{\varphi_1^{A_0}} + \sin{\gamma}\ket{{\varphi}_1^{A_0\bot}},\\
          \ket{{\varphi}_2^{A_0\bot}} = \sin{\gamma}\ket{\varphi_1^{A_0}} - \cos{\gamma}\ket{{\varphi}_1^{A_0\bot}},\\
        \end{array}
          \right.
          \\
    &\left\{
        \begin{array}{lr}
          \ket{\varphi_2^{A}} = \cos{\chi}\ket{\varphi_1^{A}} + \sin{\chi}\ket{{\varphi}_1^{A\bot}},\\
          \ket{{\varphi}_2^{A\bot}} = \sin{\chi}\ket{\varphi_1^{A}} - \cos{\chi}\ket{{\varphi}_1^{A\bot}}.\\
        \end{array}
          \right.
\end{aligned}
\end{equation}
By $\bra{\phi_1^{A_0A}}\ket{\phi_2^{A_0A}}=0$ we have
\begin{equation}\label{orthogonal}
    \cos{\alpha_1}\cos{\alpha_2}\cos{\gamma}\cos{\chi} + \cos{\alpha_1}\sin{\alpha_2}\sin{\gamma}\sin{\chi} + \sin{\alpha_1}\cos{\alpha_2}\sin{\gamma}\sin{\chi} + \sin{\alpha_1}\sin{\alpha_2}\cos{\gamma}\cos{\chi} = 0.
\end{equation}
For the effective POVM operator,
\begin{equation}
\begin{aligned}
    \tilde{M}^{A_0B_0} &= \mathrm{Tr}_{AB}\left[\left(I^{A_0}\otimes\rho^{AB}\otimes I^{B_0}\right)\left(M^{A_0A}\otimes M^{BB_0}\right)\right] \\
    &= \sum_i \mathrm{Tr}_{AB}\left[\left(I^{A_0}\otimes\ket{\Psi}\bra{\Psi} \otimes I^{B_0}\right)\left(\ket{\phi_i^{A_0A}}\bra{\phi_i^{A_0A}}\otimes M^{BB_0}\right)\right] \\
    &\equiv\sum_i \ket{\tilde{\Psi}_i}\bra{\tilde{\Psi}_i},
\end{aligned}
\end{equation}
with
\begin{equation}
\begin{aligned}
    \ket{\tilde{\Psi}_1} &=\cos{\alpha_1}\cos{\beta}\bra{\Psi}\ket{\varphi_1^A\varphi^B}\ket{\varphi_1^{A_0}\varphi^{B_0}} +\cos{\alpha_1}\sin{\beta}\bra{\Psi}\ket{\varphi_1^A{\varphi}^{B\bot}}\ket{\varphi_1^{A_0}{\varphi}^{B_0\bot}} \\
        &+\sin{\alpha_1}\cos{\beta}\bra{\Psi}\ket{{\varphi}_1^{A\bot}\varphi^B}\ket{{\varphi}_1^{A_0\bot}\varphi^{B_0}}
        +\sin{\alpha_1}\sin{\beta}\bra{\Psi}\ket{{\varphi}_1^{A\bot}{\varphi}^{B\bot}}\ket{{\varphi}_1^{A_0\bot}{\varphi}^{B_0\bot}}, \\
\end{aligned}
\end{equation}

\begin{equation}
\begin{aligned}
\ket{\tilde{\Psi}_2} =&\cos{\alpha_2}\cos{\beta}\bra{\Psi}\ket{\varphi_2^A\varphi^B}\ket{\varphi_2^{A_0}\varphi^{B_0}} +\cos{\alpha_2}\sin{\beta}\bra{\Psi}\ket{\varphi_2^A{\varphi}^{B\bot}}\ket{\varphi_2^{A_0}{\varphi}^{B_0\bot}} \\
        &+\sin{\alpha_2}\cos{\beta}\bra{\Psi}\ket{{\varphi}_2^{A\bot}\varphi^B}\ket{{\varphi}_2^{A_0\bot}\varphi^{B_0}}
        +\sin{\alpha_2}\sin{\beta}\bra{\Psi}\ket{{\varphi}_2^{A\bot}{\varphi}^{B\bot}}\ket{{\varphi}_2^{A_0\bot}{\varphi}^{B_0\bot}} \\
        =&[\cos{\alpha_2}\cos{\beta}\left(\cos{\gamma}\cos{\chi}\bra{\Psi}\ket{\varphi_1^A\varphi^B} + \cos{\gamma}\sin{\chi}\bra{\Psi}\ket{{\varphi}_1^{A\bot}\varphi^B}\right) \\&+ \sin{\alpha_2}\cos{\beta}\left(\sin{\gamma}\sin{\chi}\bra{\Psi}\ket{\varphi_1^A\varphi^B} - \sin{\gamma}\cos{\chi}\bra{\Psi}\ket{{\varphi}_1^{A\bot}\varphi^B}\right)]\ket{\varphi_1^{A_0}\varphi^{B_0}} \\+&[\cos{\alpha_2}\sin{\beta}\left(\cos{\gamma}\cos{\chi}\bra{\Psi}\ket{\varphi_1^A{\varphi}^{B\bot}} + \cos{\gamma}\sin{\chi}\bra{\Psi}\ket{{\varphi}_1^{A\bot}{\varphi}^{B\bot}}\right) \\&+ \sin{\alpha_2}\sin{\beta}\left(\sin{\gamma}\sin{\chi}\bra{\Psi}\ket{\varphi_1^A{\varphi}^{B\bot}} - \sin{\gamma}\cos{\chi}\bra{\Psi}\ket{{\varphi}_1^{A\bot}{\varphi}^{B\bot}}\right)]\ket{\varphi_1^{A_0}{\varphi}^{B_0\bot}} \\+&[\cos{\alpha_2}\cos{\beta}\left(\sin{\gamma}\cos{\chi}\bra{\Psi}\ket{\varphi_1^A\varphi^B} + \sin{\gamma}\sin{\chi}\bra{\Psi}\ket{{\varphi}_1^{A\bot}\varphi^B}\right) \\&+ \sin{\alpha_2}\cos{\beta}\left(-\cos{\gamma}\sin{\chi}\bra{\Psi}\ket{\varphi_1^A\varphi^B} + \cos{\gamma}\cos{\chi}\bra{\Psi}\ket{{\varphi}_1^{A\bot}\varphi^B}\right)]\ket{{\varphi}_1^{A_0\bot}\varphi^{B_0}}  \\+&[\cos{\alpha_2}\sin{\beta}\left(\sin{\gamma}\cos{\chi}\bra{\Psi}\ket{\varphi_1^A{\varphi}^{B\bot}} + \sin{\gamma}\sin{\chi}\bra{\Psi}\ket{{\varphi}_1^{A\bot}{\varphi}^{B\bot}}\right) \\&+ \sin{\alpha_2}\sin{\beta}\left(-\cos{\gamma}\sin{\chi}\bra{\Psi}\ket{\varphi_1^A{\varphi}^{B\bot}} + \cos{\gamma}\cos{\chi}\bra{\Psi}\ket{{\varphi}_1^{A\bot}{\varphi}^{B\bot}}\right)]\ket{{\varphi}_1^{A_0\bot}{\varphi}^{B_0\bot}}.
\end{aligned}
\end{equation}
In order that $\ket{\tilde{\Psi}_1}\sim\ket{\tilde{\Psi}_2}$ and Eq.~\eqref{orthogonal} hold, one can only have $\bra{\Psi}\ket{\varphi_1^A\varphi^B} = \bra{\Psi}\ket{{\varphi}_1^{A\bot}\varphi^B} =\bra{\Psi}\ket{\varphi_1^A{\varphi}^{B\bot}} =\bra{\Psi}\ket{{\varphi}_1^{A\bot}{\varphi}^{B\bot}} =0$, which is contradictory to that $\{\ket{\varphi^A},\ket{{\varphi}^{A\bot}}\}, \{\ket{\varphi^B},\ket{{\varphi}^{B\bot}}\}$ are complete bases. Then we can conclude that a general POVM will inevitably result in a mixed $\tilde{M}^{A_0B_0}$.
\end{proof}

\section{Proof of Lemma 3}
\begin{lemma}
For rank-one projectors $M^{A_0A}=\ketbra{\phi^{A_0A}}\in\mathcal{L}(\mathcal{H}_{A_0}\otimes\mathcal{H}_A), M^{BB_0}=\ketbra{\phi^{BB_0}}\in\mathcal{L}(\mathcal{H}_{B}\otimes\mathcal{H}_{B_0})$, where $\mathcal{H}_{A_0},\mathcal{H}_{B_0}\cong\mathcal{C}^2$, fix the measurement bases of $\mathcal{H}_{A_0},\mathcal{H}_{B_0}$ to be $\{\ket{0},\ket{1}\}_{A_0,B_0}$ and express the projectors as
\begin{equation}
\begin{aligned}
    \ket{\phi^{A_0A}} = \cos{\alpha}\ket{0} \ket{\varphi^A} + \sin{\alpha}\ket{1} \ket{\bar{\varphi}^A}, \\
    \ket{\phi^{BB_0}} = \cos{\beta}\ket{0} \ket{\varphi^B} + \sin{\beta}\ket{1} \ket{\bar{\varphi}^B}.
\end{aligned}
\end{equation}
In the condition that $\bra{\Psi}\ket{\varphi^A\bar{\varphi}^B} = \bra{\Psi}\ket{\bar{\varphi}^A\varphi^B}=0$, we have $\left|\bra{\Psi}\ket{\varphi^A\varphi^B}\right|^2+\left|\bra{\Psi}\ket{\bar{\varphi}^A\bar{\varphi}^B}\right|^2 \leq 1$. The equality is taken iff $\bra{\varphi^A}\ket{\bar{\varphi}^A}=\bra{\varphi^B}\ket{\bar{\varphi}^B}=0$, and the subspaces of the two parties are spanned by $\left\{\ket{\varphi^A},\,\ket{\bar{\varphi}^A}\right\},\, \left\{\ket{\varphi^B},\,\ket{\bar{\varphi}^B}\right\}$, respectively. In particular, no restriction on the dimensions of $\mathcal{H}_{A},\mathcal{H}_{B}$ is made.
\end{lemma}
\begin{proof}
In the proof for this lemma, we do not restrict the dimensions of $\mathcal{H}_A,\mathcal{H}_B$. Let
\begin{equation}
\begin{aligned}
    &\ket{\bar{\varphi}^A} = \cos{\theta}\ket{\varphi^A} + \sin{\theta}\ket{{\varphi}^{A\bot}},\\
    &\ket{\bar{\varphi}^B} = \cos{\gamma}\ket{\varphi^B} + \sin{\gamma}\ket{{\varphi}^{B\bot}},
\end{aligned}
\end{equation}
where $\ket{{\varphi}^{A\bot}},\,\ket{{\varphi}^{B\bot}}$ are some states that lay orthogonally to $\ket{\varphi^A},\,\ket{\varphi^B}$ in the subspaces of Alice and Bob, respectively. Then we have
\begin{equation}
    \bra{\Psi}\ket{\bar{\varphi}^A \bar{\varphi}^B} = \cos{\theta}\cos{\gamma}\bra{\Psi}\ket{\varphi^A \varphi^B} + \cos{\theta}\sin{\gamma}\bra{\Psi}\ket{\varphi^A {\varphi}^{B\bot}} + \sin{\theta}\cos{\gamma}\bra{\Psi}\ket{{\varphi}^{A\bot} \varphi^B} + \sin{\theta}\sin{\gamma}\bra{\Psi}\ket{{\varphi}^{A\bot} {\varphi}^{B\bot}}.
\end{equation}
Since $\bra{\Psi}\ket{\varphi^A\bar{\varphi}^B}=0$, we have
\begin{equation}\label{condition1}
    \cos{\gamma}\bra{\Psi}\ket{\varphi^A \varphi^B} + \sin{\gamma}\bra{\Psi}\ket{\varphi^A {\varphi}^{B\bot}}=0.
\end{equation}
Similarly,
\begin{equation}\label{condition2}
    \cos{\theta}\bra{\Psi}\ket{\varphi^A \varphi^B} + \sin{\theta}\bra{\Psi}\ket{{\varphi}^{A\bot} \varphi^B}=0.
\end{equation}
Therefore, we derive the following equation
\begin{equation}
    \bra{\Psi}\ket{\bar{\varphi}^A \bar{\varphi}^B} = -\cos{\theta}\cos{\gamma}\bra{\Psi}\ket{\varphi^A \varphi^B} + \sin{\theta}\sin{\gamma}\bra{\Psi}\ket{{\varphi}^{A\bot} {\varphi}^{B\bot}}.
\end{equation}
On the other hand, required by Eq.~\eqref{condition1}\eqref{condition2}, as long as $\ket{\varphi^A}\neq\ket{\bar{\varphi}^A},\, \ket{\varphi^B}\neq\ket{\bar{\varphi}^B}$, there are following relations
\begin{equation}
\begin{aligned}
    &\bra{\Psi}\ket{\varphi^A {\varphi}^{B\bot}} = -\frac{\cos{\gamma}}{\sin{\gamma}}\bra{\Psi}\ket{\varphi^A \varphi^B},\\
    &\bra{\Psi}\ket{{\varphi}^{A\bot} \varphi^B} = -\frac{\cos{\theta}}{\sin{\theta}}\bra{\Psi}\ket{\varphi^A \varphi^B}.
\end{aligned}
\end{equation}
As $\ket{\varphi^A \varphi^B},\,\ket{\varphi^A {\varphi}^{B\bot}},\,\ket{{\varphi}^{A\bot} \varphi^B},\,\ket{{\varphi}^{A\bot} {\varphi}^{B\bot}}$ are orthogonal to each other, the following inequality holds
\begin{equation}\label{subcompleteness}
\begin{aligned}
    \left|\bra{\Psi}\ket{{\varphi}^{A\bot} {\varphi}^{B\bot}}\right|^2 &\leq 1- \left|\bra{\Psi}\ket{\varphi^A \varphi^B}\right|^2 - \left|\bra{\Psi}\ket{\varphi^A {\varphi}^{B\bot}}\right|^2 - \left|\bra{\Psi}\ket{{\varphi}^{A\bot} \varphi^B}\right|^2\\
    &= 1-\left(1+\frac{\cos^2{\gamma}}{\sin^2{\gamma}}+\frac{\cos^2{\theta}}{\sin^2{\theta}}\right) \left|\bra{\Psi}\ket{\varphi^A \varphi^B}\right|^2,
\end{aligned}
\end{equation}
and thereafter
\begin{equation}
\begin{aligned}
    \left|\bra{\Psi}\ket{\bar{\varphi}^A \bar{\varphi}^B}\right|^2 =& \cos^2{\theta}\cos^2{\gamma}\left|\bra{\Psi}\ket{\varphi^A \varphi^B}\right|^2 + \sin^2{\theta}\sin^2{\gamma}\left|\bra{\Psi}\ket{{\varphi}^{A\bot} {\varphi}^{B\bot}}\right|^2 - 2\sin{\theta}\cos{\theta}\sin{\gamma}\cos{\gamma}\bra{\Psi}\ket{\varphi^A \varphi^B}\bra{\Psi}\ket{{\varphi}^{A\bot} {\varphi}^{B\bot}}\\
    \leq&\cos^2{\theta}\cos^2{\gamma}\left|\bra{\Psi}\ket{\varphi^A \varphi^B}\right|^2 + \sin^2{\theta}\sin^2{\gamma}\left[1-\left(1+\frac{\cos^2{\gamma}}{\sin^2{\gamma}}+\frac{\cos^2{\theta}}{\sin^2{\theta}}\right) \left|\bra{\Psi}\ket{\varphi^A \varphi^B}\right|^2\right]
    \\&+ 2\sin{\theta}\cos{\theta}\sin{\gamma}\cos{\gamma}\left|\bra{\Psi}\ket{\varphi^A \varphi^B}\right|\sqrt{1-\left(1+\frac{\cos^2{\gamma}}{\sin^2{\gamma}}+\frac{\cos^2{\theta}}{\sin^2{\theta}}\right) \left|\bra{\Psi}\ket{\varphi^A \varphi^B}\right|^2},
\end{aligned}
\end{equation}
thus
\begin{equation}
\begin{aligned}
    &\left|\bra{\Psi}\ket{\varphi^A \varphi^B}|^2 + |\bra{\Psi}\ket{\bar{\varphi}^A \bar{\varphi}^B}\right|^2\\
    \leq&\sin^2{\theta}\sin^2{\gamma} + (1+\cos^2{\theta}\cos^2{\gamma}-\sin^2{\theta}\sin^2{\gamma}-\sin^2{\theta}\cos^2{\gamma}-\cos^2{\theta}\sin^2{\gamma})\left|\bra{\Psi}\ket{\varphi^A \varphi^B}\right|^2\\
    &+2\sin{\theta}\cos{\theta}\sin{\gamma}\cos{\gamma}\left|\bra{\Psi}\ket{\varphi^A \varphi^B}\right|\sqrt{1-\left(1+\frac{\cos^2{\gamma}}{\sin^2{\gamma}}+\frac{\cos^2{\theta}}{\sin^2{\theta}}\right)\left|\bra{\Psi}\ket{\varphi^A \varphi^B}\right|^2}\\
    =&\sin^2{\theta}\sin^2{\gamma}+2\cos^2{\theta}\cos^2{\gamma}\left|\bra{\Psi}\ket{\varphi^A \varphi^B}\right|^2\\
    &+2\sin{\theta}\cos{\theta}\sin{\gamma}\cos{\gamma}\left|\bra{\Psi}\ket{\varphi^A \varphi^B}\right|\sqrt{1-\left(1+\frac{\cos^2{\gamma}}{\sin^2{\gamma}}+\frac{\cos^2{\theta}}{\sin^2{\theta}}\right)\left|\bra{\Psi}\ket{\varphi^A \varphi^B}\right|^2}.
\end{aligned}
\end{equation}
For the right side of this inequality, it can be regarded as a function of the form $f(x) = ax^2 + bx\sqrt{1-cx^2}+d$, with the variant
\begin{equation}
    x = \left|\bra{\Psi}\ket{\varphi^A \varphi^B}\right|\in[0,1],
\end{equation}
and the coefficients
\begin{equation}
\begin{aligned}
    \left\{
        \begin{array}{lr}
          a = 2\cos^2{\theta}\cos^2{\gamma}, &\\
          b = 2\sin{\theta}\cos{\theta}\sin{\gamma}\cos{\gamma}, &\\
          c = 1+\dfrac{\cos^2{\gamma}}{\sin^2{\gamma}}+\dfrac{\cos^2{\theta}}{\sin^2{\theta}}, &\\
          d = \sin^2{\theta}\sin^2{\gamma}.
        \end{array}
          \right.
\end{aligned}
\end{equation}
In our problem, $c\geq1$, hence it is required that $x^2\geq\frac{1}{c}$. The derivative of $f(x)$ is
\begin{equation}
    f'(x) = 2ax+b\sqrt{1-cx^2}-\frac{bcx^2}{1-cx^2}.
\end{equation}
Let $f'(x)=0$ and we have
\begin{equation}
    x^2 = \frac{1}{2c}\pm\frac{a}{2c}\frac{1}{\sqrt{a^2+b^2c}},
\end{equation}
where the term $\frac{a}{2c}\frac{1}{\sqrt{a^2+b^2c}}\leq\frac{a}{2c}\frac{1}{a}=\frac{1}{2c}$, so we are able to take the value $x=\pm\sqrt{\frac{1}{2c}\pm\frac{a}{2c}\frac{1}{\sqrt{a^2+b^2c}}}$. By some calculation with respect to $f'(x)$, for the maximum value of $f(x)$, we only need to consider its value at the points $x=\sqrt{\frac{1}{c}}$ and $x=\sqrt{\frac{1}{2c}-\frac{a}{2c}\frac{1}{\sqrt{a^2+b^2c}}}$.\\

\emph{(1) $x=\sqrt{\frac{1}{c}}$}\\
The function value at this point
\begin{equation}
\begin{aligned}
    f(x) &= \frac{a}{c}+d\\
    &=\frac{2\cos^2{\theta}\cos^2{\gamma}}{1+\frac{\cos^2{\gamma}}{\sin^2{\gamma}}+\frac{\cos^2{\theta}}{\sin^2{\theta}}}+\sin^2{\theta}\sin^2{\gamma}\\
    &=\frac{1+\cos^2{\theta}\cos^2{\gamma}}{1+\frac{\cos^2{\gamma}}{\sin^2{\gamma}}+\frac{\cos^2{\theta}}{\sin^2{\theta}}}\\
    &\leq 1,
\end{aligned}
\end{equation}
and the equality is taken iff $\cos{\theta}=\cos{\gamma}=0$. This indicates that $\bra{\varphi^A}\ket{\bar{\varphi}^A} = \bra{\varphi^B}\ket{\bar{\varphi}^B}=0$.

\emph{(2) $x=\sqrt{\frac{1}{2c}-\frac{a}{2c}\frac{1}{\sqrt{a^2+b^2c}}}$}\\
The function value at this point
\begin{equation}
\begin{aligned}
    f(x) &= ax^2 + bx\sqrt{1-(\frac{1}{2}-\frac{a}{2}\frac{1}{\sqrt{a^2+b^2c}})}+d\\
    &= \frac{a}{2c} + \frac{1}{2c}\frac{b^2c-a^2}{\sqrt{b^2c+a^2}}+d\\
    &= \frac{1+\cos{\theta}\cos{\gamma}-2\cos^3{\theta}\cos^3{\gamma}}{1+\frac{\cos^2{\gamma}}{\sin^2{\gamma}}+\frac{\cos^2{\theta}}{\sin^2{\theta}}}\\
    &\leq 1,
\end{aligned}
\end{equation}
and the equality is taken iff $\cos{\theta}=\cos{\gamma}=0$.

Above all, we come to the conclusion that when $\bra{\Psi}\ket{\varphi^A \bar{\varphi}^B} = \bra{\Psi}\ket{\bar{\varphi}^A \varphi^B}=0$, the inequality $\left|\bra{\Psi}\ket{\varphi^A \varphi^B}\right|^2 + \left|\bra{\Psi}\ket{\bar{\varphi}^A \bar{\varphi}^B}\right|^2 \leq 1$ holds, and equality is taken iff $\bra{\varphi^A}\ket{\bar{\varphi}^A} = \bra{\varphi^B}\ket{\bar{\varphi}^B}=0$, and the subspaces of Alice and Bob are spanned by $\left\{\ket{\varphi^A},\,\ket{\bar{\varphi}^A}\right\}, \,\left\{\ket{\varphi^B},\,\ket{\bar{\varphi}^B}\right\}$, respectively (notice the inequality scaling~\eqref{subcompleteness}).
\end{proof}




\end{document}